\documentclass{article}
\usepackage{amsmath, amssymb, amsthm, amscd, MnSymbol}
\usepackage{mathtools}
\usepackage{spconf}
\usepackage{epsfig}
\usepackage{color}
\usepackage{widetext} %% to add one column equation in the middel of the paper

\newtheorem{theorem}{Theorem}[section]
\newtheorem{lemma}[theorem]{Lemma}
\def\w{{\bf w}}
\def\s{{\bf s}}

\def\h{{\bf h}}
\def\x{{\bf x}}

\def\x{{\mathbf x}}

\title{Optimal non-coherent data detection for massive SIMO wireless systems: A polynomial complexity solution}
\twoauthors
  {Haider Ali Jasim Alshamary, Weiyu Xu } %\sthanks{.....}
   {University of Iowa\\
   Electrical and Comptuer Engineering\\
   Iowa City, USA}
  {Tareq Al-Naffouri, Alam Zaib }
   {King Abdullah Univ. of Science and Technology\\
   Electrical Engineering Program\\
   Saudi Arabia}
\begin{document}
\ninept
\renewcommand{\baselinestretch}{1}
\maketitle
\begin{abstract}
 Massive MIMO systems have made significant progress in increasing spectral and energy efficiency over traditional MIMO systems by exploiting large antenna arrays.  In this paper we consider the joint maximum likelihood (ML) channel estimation and data detection problem for massive SIMO  (single input multiple output) wireless systems. Despite the large number of unknown channel coefficients for massive SIMO systems, we improve an algorithm to achieve the exact ML non-coherent data detection with a low expected complexity. We show that the expected computational complexity of this algorithm is linear in the number of receive antennas and polynomial in channel coherence time. Simulation results show the performance gain of the optimal non-coherent data detection with a low computational complexity.

\end{abstract}
\section{Introduction}
\label{sec:intro}
Using multiple-antenna arrays has been well known  for its benefits: high reliability, high spectral efficiency and interference reduction. Recently a new approach, \textit{massive} MIMO, has emerged by equipping communication terminals  with a huge number of antennas. This reaps the benefits of the traditional MIMO systems on a much larger scale. \cite{Thomas} mathematically showed that the effect of fast fading and non-correlated noise is eliminated as the number of receive antennas approaches infinity. Since then, extensive research interests have been generated in massive MIMO. For example, massive MIMO systems' information-theoretic and propagation aspects are discussed in \cite{ScalingMIMO,MUMIMO}. Research on massive MIMO has also focused on many other aspects, including  transmit and receive schemes, the effect of pilot contamination,  energy efficiency,  and channel estimation for massive MIMO systems, as overviewed in \cite{Challenges,NextGeneration}.

 Knowledge of the channel state information (CSI) is required to achieve the advantages of massive MIMO systems \cite{Challenges}. However, accurately estimating the channel gains in wireless systems is a big challenge, especially in fast fading environments \cite{HowmanyAntenna}. In case of conventional MIMO systems, differential modulation techniques, blind and semi-blind, and pilot based algorithms are used to solve the problem of channel tracking \cite{Stoica, HarisConference, Ma, Swindlehurst, Manton}. Although these algorithms have improved the performance of non-coherent MIMO systems, they are not optimal for massive time-varying channels.  Compared with traditional MIMO systems,  it is even more challenging to perform accurate channel state estimation for massive MIMO systems, considering massive MIMO's large number of unknown channel coefficients. It is of great theoretical and practical interest to investigate near-optimal or optimal non-coherent data transmission and data detection schemes for massive MIMO systems \cite{NextGeneration}.

In this paper, we consider the problem of joint ML channel estimation and data detection for massive SIMO systems. An extensive list of works have addressed non-coherent data detection problems for conventional SIMO wireless systems or wireless systems in general. Most existing efficient MIMO non-coherent signal detection algorithms are suboptimal compared with the exact ML algorithms.  However, there are a few exceptions. For instance, sphere decoder algorithm was used in \cite{Hassibi} to solve the joint ML non-coherent problem for SIMO wireless systems. Sphere decoder algorithm reduces the computational complexity by restricting the ML detection search to a subset of the signal space.   \cite{Ma} also used sphere decoder algorithms to achieve the ML channel estimation and data detection for orthogonal space time block coded (OSTBC) wireless systems. In \cite{Hassibi} and \cite{Ma}, the sphere decoder algorithm has been shown as an exact ML non-coherent detection algorithm which has a lower complexity than the exhaustive search,  but the sphere decoder works only for constant-modulus  constellations. \cite{Weiyu} proposed an exact joint ML channel estimation and signal detection algorithm for SIMO systems with general constellations.  \cite{OFDMBlind} proposed an exact ML channel estimation and data detection for OFDM wireless systems with general constellations.  An ML non-coherent signal detection algorithm for OSTBC was developed in \cite{Dimitris} for constant-modulus constellations. The algorithm proposed in \cite{Dimitris} uses recent results on efficient maximization of reduced-rank quadratic form to achieve a polynomial complexity.

 The optimal non-coherent data detection algorithms from \cite{Hassibi} and \cite{Weiyu} did not look at the non-coherent data detection complexity as the number of receive antennas grows large in massive SIMO systems. Furthermore, the algorithm in \cite{Dimitris} gives an exact ML solution only when the matrix in the related quadratic form optimization problem has low rank, but  this  low-rank assumption does not hold for massive SIMO systems with a large number of receive antennas.  Without efficient algorithms achieving optimal non-coherent data detection for massive MIMO systems, it was not known how suboptimal non-coherent data detection methods compare with the ML non-coherent data detection methods. It was believed that the goal of achieving joint ML channel estimation and data detection is even more difficult for massive MIMO systems, because of a large number of unknown channel coefficients \cite{ScalingMIMO}.

In this paper we study and improve a joint ML channel estimation and data detection algorithm for massive SIMO systems. Surprisingly, despite a large number of unknown channel coefficients for massive SIMO systems, this algorithm achieves the exact ML non-coherent data detection with a low expected complexity. We theoretically show that the expected computational complexity of the algorithm is linear in the number of receive antennas and polynomial in channel coherence time. Simulation results demonstrate the performance gain of our optimal non-coherent data detection with a low computational complexity. To the best of our knowledge, for the first time, we have demonstrated the exact performance gap between the optimal non-coherent data detection algorithm, and suboptimal non-coherent data detection algorithms, for massive SIMO systems.

The rest of this paper is organized as follows. Section \ref{sec:problem} sets up the system model and presents the ML non-coherent data detection algorithm. Section \ref{sec:complexity} derives the expected complexity of the algorithm. Simulation results are provided and discussed in Section \ref{sec:simulation}. Section \ref{sec:conclusion} concludes our paper and highlights our contribution.

\section{The JOINT CHANNEL ESTIMATION AND SIGNAL DETECTION PROBLEM}
\label{sec:problem}
Let $T$ denote the length of a data packet during which the channel remains constant. The channel output for a SIMO system with $N$ receive antennas is given by
\begin{equation}\label{system}
X=\h\s^*+W,
\end{equation}
where $\h \in \mathcal{C}^{N \times 1}$ is the SIMO channel vector,
$\s^* \in \mathcal{C}^{1 \times T}$ is the transmitted symbol
sequence, and $W \in \mathcal{C}^{N \times T}$ is an additive noise
matrix whose elements are assumed to be i.i.d. complex Gaussian
random variables. We also assume the entries of $\s^*$ are
i.i.d. symbols from a certain constant-modulus constellation $\Omega$ (such as BPSK or QPSK) which has unit expected energy, i.e.,
\begin{equation}
E(|\mathbf{s}_{k}|^2)=1, k=1,2,...,T. \label{Power}
\end{equation}
We assume $\h$ as a deterministic unknown channel with no priori information known about it \cite{Stoica}\cite{Ma}.
Then, the joint ML channel estimation and data detection problem for SIMO systems is given by the following mixed optimization problem \vspace{-0.05in}
\begin{equation}
\min_{\h, \s^* \in \Omega^T}\| X-\h\s^*\|^2, \label{eq:mixed}
\end{equation}
where $\Omega^T$ denotes the set of $T$-dimensional signal vectors. From \cite{Hassibi}, the optimization to (\ref{eq:mixed}) over $\h$ is a least square problem while the optimization over $\s^*$ is an integer least square problem, since each elements of $\s^*$ is chosen from a fixed constellation $\Omega$. By \cite{HarisConference}, for any given symbol vectors $\s^*$, the channel vector $\h$ that minimizes (\ref{eq:mixed}) is
\begin{equation}
\hat{\h}=X\s (\s^*\s)^{-1}=X\s/\|\s\|^2, \label{eq:opth}
\end{equation}
Substituting (\ref{eq:opth}) into (\ref{eq:mixed}), we get
\begin{equation}
\|  X(\underbrace { I-\frac{1}{\|\s\|^2}\s\s^*) }_{=P_{\s}}
\|^2=\text{tr}(XP_{\s}X^*
)=\text{tr}(XX^*)-\frac{1}{\|\s\|^2}\s^*X^*X\s, \label{eq:optmetric}
\end{equation}
As pointed ou t in \cite{HarisConference}, if the modulation constellation is of constant modulus (such as
QPSK), the minimization of (\ref{eq:optmetric}) over $\s^*$ is
equivalent to solve the following problem:
\begin{equation}
\max_{\s^* \in \Omega^T}\s^*X^*X\s, \label{eq:max}
\end{equation}
The quadratic form in (\ref{eq:max}) for a constant modulus modulation can be changed into an equivalent minimization problem by using the maximum eigenvalue of $X^*X$. Thus, (\ref{eq:max}) can be represented as
\begin{equation}
\min_{\s \in \Omega^T}\s^*\underbrace{(\rho I-\frac{X^*X}{N}}_{=\Im})\s, \label{eq:min}
\end{equation}
where $\rho$ is a slightly larger value than the maximum eigenvalue of $\frac{X^*X}{N}$. The traditional solution of an integer least square optimization problem in (\ref{eq:min}) is by using exhaustive search over the entire lattice. However, the computational complexity of exhaustive search is exponential in $T$. Sphere decoder was used in \cite{HarisConference} to efficiently solve (\ref{eq:min}) with a lower computational complexity than exhaustive search. Instead of searching over all the hypotheses in the lattice, sphere decoder attempts to look at the lattice points within a radius $r$. As a result, the searching process of sphere decoder only visits the sequences that are inside the hypersphere of radius $r$ %\vspace{-0.1 in}
\begin{equation}
\s^*(\rho I-\frac{X^*X}{N})\s \leq r^2. \label{eq:spheresearch}
\end{equation}

From the way in which $\rho$ is determined, the matrix $\Im$ in (\ref{eq:min}) is positive semidefinite. We can use the Cholesky decomposition to factorize $\Im$ as
\begin{equation}
\Im=R^*R,
\end{equation}
where $R$ is an upper triangular matrix. Now we can rewrite (\ref{eq:min}) as
\begin{align}
\label{eq:min10}
  \min_{\s^* \in \Omega^T}\s^*(\rho I-\frac{X^*X}{N})\s &= \min_{\s^* \in \Omega^T}\s^*R^*R\s\notag\\
&=\min_{\s^* \in \Omega^T}\|\ R \s \|^2.
\end{align}
Since $R$ is an upper triangular matrix, $R\s$ can be expanded as
\begin{equation}\label{eq:Metric}
M_{\s^*}=\sum^T_{i=1} \|\sum_{k=i}^{T} L_{i,k} \s_k\|^2,
\end{equation}
where $M_{\s^*}$ is the metric of the transmitted vector $\s^*$, and $L_{i,k}$ is the entry of $R$ in the $i$-th row and $j$-th column. %We will consider layer $T$ as the first layer in the tree search of ML algorithm.
For each $i$ between $1$ and $T$, we further define
\begin{equation}\label{eq:Mertici1}
 M_{\s^*_{i:T}}= \|\sum_{k=i}^{T} L_{i,k} \s_k\|^2+M_{\s^*_{i+1:T}},
\end{equation}
where the partial sequence $\s^*_{i:T}$ consist of elements $\s^*_i$, $\s^*_{i+1}$, ..., $\s^*_{T}$, $M_{\s^*_{i:T}}$ is the metric of the partial sequence $\s^*_{i:T}$, and $M_{\s^*_{T+1:T}}=0$ by default. If the set of possible data sequences are represented in a tree structure as in \cite{HarisConference}, we refer to $\s^*_{i:T}$ as a layer-$i$ node in the tree. Now we present the algorithm from \cite{HarisConference} for joint ML channel estimation and data detection.

\noindent \emph{\textbf{Joint ML channel estimation data detection algorithm}} \\
Input: radius $r$, matrix $R$, constellation $\Omega$ and a $1 \times T$
index vector $I$
\begin{enumerate}
\item Set $i=T$, $r_{i}=r$, $I(i)=1$ and set $\s^*_{i}=\Omega(I(i))$.
\item (Computing the bounds) Compute the metric $M_{\s^*_{i:T}}$. If
$M_{\s^*_{i:T}}>r^2$, go to 3; else, go to 4;
\item (Backtracking) Find the smallest $i\leq j \leq T$ such
that $I(j)<|\Omega|$. If there exists such $j$, set $i=j$ and go to
5; else go to 6.

\item If $i=1$, store current $\s^*$, update $r^2=M_{\s^*_{i:T}}$ and go to 3; else set $i=i-1$, $I(i)=1$ and
$\s^*_{i}=\Omega(I(i))$, go to 2.

\item  Set $I(i)=I(i)+1$ and $s^*_{i}=\Omega(I(i))$.
Go to 2.

\item If any sequence $\s^*$ is ever found in Step 4, output the latest
stored full-length sequence as the ML solution; otherwise, double $r$
and go to 1.\\
\end{enumerate}

In our analysis of this algorithm for massive SIMO systems, we will slightly change the algorithm in the last step: if no sequence is ever found in Step 4, we will increase $r$ to $\infty$.

\subsection{Choice of radius $r$}
The choice of the radius $r$ has a big influence on the complexity of this ML algorithm. If $r^2$ is chosen bigger than the metric of any sequence $\tilde{\s}\in|\Omega|^T$, the ML algorithm will visit all the tree nodes under that radius. If $r^2$ is too small, then the ML sequence may be outside the search radius, and the ML algorithm will have to search again under a new larger radius.

 \cite{HarisConference, SphereComplexity} derived how to choose $r$ such that with a certain probability, the transmitted sequence has a metric no bigger than $r^2$. However, the radius choice in \cite{HarisConference} is for a fixed number of receive antennas, and  for high signal-to-noise ratio (SNR).

In this paper, we quantify the choice of radius $r$ when the number of receive antennas is big, as in massive MIMO systems. In fact we set $r^2$ as any constant $c$ such that
$$r^2 =c<  \frac{T}{2}.$$

We remark that this radius choice is different from \cite{HarisConference}.  More specifically, the new radius value does not depend on the SNR or the number of receive antennas.  In fact, one can choose the radius of $r^2$ to be a positive constant arbitrarily close to 0, for a large SIMO system. In the next section, we will show that, under this new radius, the joint ML channel estimation and data detection algorithm has very low computational complexity.

\section{Algorithm computational complexity}
\label{sec:complexity}
The computational complexity of the joint ML channel estimation and data detection algorithm for SIMO systems is mainly determined by the number of visited nodes in each layer. By ``visited nodes'', we mean the partial sequences $\s^*_{i:T}$ for which $M_{\s^*_{i:T}}$ is computed in the algorithm. The fewer the visited nodes, the lower computational complexity the joint ML algorithm needs. In this section, we will show that the number of visited points in each layer will converge to a constant number for a sufficient large number of receive antennas.

\begin{theorem}
In the joint maximum-likelihood joint channel estimation and data detection algorithm, the expected number of visited points at layer $i$  with $N$ receive antennas converges to $|\Omega|$ for $i\leq (T-1)$, as $N$ goes to infinity. The joint ML algorithm only visits one tree node at layer $i=T$.
\label{thm:ltt1}
\end{theorem}

\begin{proof}[Proof of Theorem \ref{thm:ltt1}]
The number of visited nodes at layer $i$ ($1\leq i \leq T-1$) in the joint ML algorithm is equal to $|\Omega|$,  if there is one and only one tree node $\widetilde{\s}^*_{(i+1):T}$ such that $M_{\widetilde{\s}^*_{(i+1):T}} \leq r^2$. In fact, we will prove that, the transmitted  $\s^*_{(i+1):T}$ will be the only sequence satisfying $M_{\widetilde{\s}^*_{(i+1):T}} \leq r^2$, with high probability as the number of receive antennas $N \rightarrow \infty$. To prove this, we first show this conclusion is true for the average case with $\Im_{E}=\rho_{E} I-\frac{E[X^*X]}{N}$, where $\rho_{E}$ is the maximum eigenvalue of $\frac{E[X^*X]}{N}$. Then we use the concentration results for $\frac{X^*X}{N}$ to prove  that, for $\Im=\rho I-\frac{E[X^*X]}{N}$,  the transmitted  $\s^*_{(i+1):T}$ will also be the only sequence satisfying $M_{s^*_{(i+1):T}} \leq r^2$.

For the average case, we first derive $E[X^*X]$, and factorize $\rho_{E} I-\frac{E[X^*X]}{N}$ using the Cholesky decomposition. Using the upper triangular matrix generated from the Cholesky decomposition, we show that the transmitted $\s^*_{(i+1):T}$ will be the only sequence satisfying $M_{\s^*_{(i+1):T}} \leq r^2$ under $\Im=\rho_{E} I-\frac{E[X^*X]}{N}$.

We can write (\ref{system}) as
\begin{align}
   [\x_{1} \;  \x_{2} \;  \cdot \;  \cdot \;  \x_{T}] &= [\s^*_{1}\h \; \s^*_{2}\h \; \cdot  \; \cdot \; \s^*_{T}\h]+[\w_{1} \; \w_{2} \; \cdot \; \cdot \; \w_{T}]\notag\\
   &= [\s^*_{1}\h+\w_{1} \;\; \s^*_{2}\h+\w_{2} \;\; \cdot  \;\; \cdot \;\; \s^*_{T}\h+\w_{T}],\notag\\ \nonumber
\end{align}
where $\x_{i}$ is the $i$-th column vector of $X$. Then $E[X^*X]$ is equal to
\begin{equation}
E \left\{\begin{bsmallmatrix}(\s^*_{1}\h+\w_{1})^* \\ (\s^*_{2}\h+\w_{2})^* \\ \vdots \\ (\s^*_{T}\h+\w_{T})^*\end{bsmallmatrix}
\begin{bsmallmatrix} (\s^*_{1}\h+\w_{1}) & (\s^*_{2}\h+\w_{2}) & \cdots & (\s^*_{T}\h+\w_{T}) \end{bsmallmatrix} \right\}\notag\\ \nonumber.
\end{equation}
Since the entries of $\h$ are independent complex Gaussian random variables with unit variance and zero mean, $E[\h^* \h]= E[\sum^N_{i=1} h_{i}^*h_{i}]%= \sum^N_{i=1}E[h_{i}^*h_{i}]%
=N$.  After some algebra,  we can rewrite $E[X^*X]/N$ as
\begin{equation}
\begin{bmatrix}
	\s_{1}\s^*_{1}+\sigma^2_{w} & \s_{1}\s^*_{2} & \cdots & \s_{1}\s^*_{T} \\
    \s_{2}\s^*_{1} & \s_{2}\s^*_{2}+\sigma^2_{w} & \cdots & \s_{2}\s^*_{T} \\
    \vdots & \vdots & \ddots & \vdots \\
    \s_{T}\s^*_{1} & \s_{T}\s^*_{2}  &  \cdots & \s_{T}\s^*_{T}+\sigma^2_{w} \\
    \end{bmatrix} \label{X*X}.
\end{equation}

Obviously (\ref{X*X}) is a Hermitian matrix with a full column rank. The maximum eigenvalue of $\frac{E[X^*X]}{N}$ is $\rho_{E}=T+\sigma^2_{w}$. Then we can write $A=\rho_{E} I- \frac{E[X^*X]}{N}$ as
\begin{equation}
A=
\begin{bmatrix}
			T-\s_{1}\s^*_{1} & -\s_{1}\s^*_{2} & \cdots & -\s_{1}\s^*_{T} \\
             -\s_{2}\s^*_{1} & T-\s_{2}\s^*_{2} & \cdots & -\s_{2}\s^*_{T}\\
             \vdots & \vdots & \vdots &\vdots \\
             -\s_{T}\s^*_{1} & -\s_{T}\s^*_{2}& \cdots & T-\s_{T}\s^*_{T} \\
             \end{bmatrix} \nonumber.
\end{equation}
We then decompose $(\rho_{E} I- \frac{E[X^*X]}{N})$ into $\grave{R}^*\grave{R}$ using the Cholesky decomposition in \cite{NumbericalMathmatics}. Then we have
 \begin{equation}
  \grave{R}=\begin{bmatrix} L_{1,1} & L_{1,2}  & L_{1,3}  & \cdot  & \cdot & \ L_{1,T}\\
                    0&L_{2,2}  &L_{2,3}  &\cdot  &\cdot  & L_{2,T}\\
                    0 & 0  &L_{3,3}  &\cdot  &\cdot  & L_{3,T}\\
                    0&0  &0 &\cdot &\cdot & L_{T,T} \notag\end{bmatrix},
 \end{equation}
where $L^*_{i,i}= \sqrt{a_{i,i}-\sum^{i-1}_{k=1} L_{k,i}L^*_{k,i}}$, $L^*_{i,j}=\frac{1}{L_{i,i}}( a_{j,i}-\sum^{i-1}_{k=1}L_{k,i} L^*_{k,j})$ for $1\leq i<j\leq T$, and $a_{i,j}$ is the entry of $(\rho_{E} I- \frac{E[X^*X]}{N})$ with row index $i$, and column index $j$. Thus $\grave{R}$ is given by (\ref{eq:Rmat}) (listed on the top of next page).
\begin{figure*}
%\hrule
	\begin{equation}
  		\grave{R}=\begin{bmatrix} \sqrt{T-1} & \frac{-(\s_{1}\s^*_{2})}{\sqrt{T-1}}  & \frac{-(\s_{1}\s^*_{3})}{\sqrt{T-1}}  & \cdots & \frac{-(\s_{1}\s^*_{T})}{\sqrt{T-1}}\\
                    0 & \sqrt{T-1 -\frac{1}{T-1}}  & \frac{1}{L_{2,2}} \left [-(\s_{2}\s^*_{3})- \frac{(\s_{2}\s^*_{3})}{T-1}\right]  & \cdots  & \frac{1}{L_{2,2}} \left [-(\s_{2}\s^*_{T})- \frac{(\s_{2}\s^*_{T})}{T-1}\right]\\
                    0 & 0  &\sqrt{T-1 -\frac{1}{T-1}-\frac{T}{(T-1)(T-2)}}  & \cdots  & \frac{1}{L_{3,3}} \left [-(\s_{3}\s^*_{T})- \frac{(\s_{3}\s^*_{T})}{T-1}-\frac{(\s_{3}\s^*_{T})T}{(T-1)(T-2)}\right]\\
                    \vdots & \vdots & \vdots & \ddots & \vdots \\
                    0      &   0    &   0    & \cdots & \sqrt{T-1 -\frac{1}{T-1}-\cdot-\frac{T}{(T-(T-2))(T-(T-1))}}\end{bmatrix}.
    \label{eq:Rmat}
	\end{equation}
\hrule
\end{figure*}

We can see that $L_{ii}=\sqrt{(T-1)-\sum_{j=1}^{i-1} \frac{T}{(T-(j-1))(T-j)}}$ for $1< i \leq T$. Now we can use $ \grave{R}$ in (\ref{eq:Rmat}) as the upper triangular matrix of Cholesky decomposition to solve the minimization equation in (\ref{eq:min10}). In fact, based on (\ref{eq:Metric}), the metric $M_{\s^*_{1:T}} (\grave{R})$ from (\ref{eq:min}) is
\begin{align}\label{eq:Mertici3}
M_{\s^*_{1:T}}=\s^*A\s &=\s^*(TI-\s^*\s)\s \notag\\
  &=T\s^*\s-\s^*\s^*\s\s\notag\\
  &=T^2-T^2 \notag\\
  &=0,
\end{align}
since $\s^* \s=T$. Because $M_{\s^*}=\sum^T_{i=1} \|\sum_{k=i}^{T} L_{i,k} \s_k\|^2$, from (\ref{eq:Mertici3}), we must have $\| \sum_{k=i}^{T} L_{i,k}$ $ \s_k\|^2=0$ for every $1\leq i \leq T$. This in turn implies that $M_{\s^*_{i:T}}=0$, and $\sum_{k=i}^{T} L_{i,k} \s_k=0$ for every $1\leq i \leq T$.  On the other hand, according to Lemma \ref{thm:ltt2}, for any other $\widetilde{\s}\neq \s$, $M_{\widetilde{\s}^*_{i:T}} \neq 0$,  where $i$ is the integer closest to $T$ such that $\s^*_i \neq \widetilde{\s}^*_{i}$.

When $i=T$,  the joint ML algorithm  will visit only $1$ tree node, namely $\s^*_{T}$, whose metric is equal to $0$,  because $\s^*_{T}$ is predetermined to resolve phase ambiguity; when $i<T$, at layer $i$, we also only have one sequence $\widetilde{\s}^*_{i:T}=\s^*_{i:T}$ such that $M_{\widetilde{\s}^*_{i:T}}= 0$. This will prove Theorem \ref{thm:ltt1}, under the assumption that $X^*X=E[X^*X]$.

Now we prove that, with high probability, $X^*X/N$ is close to $E[X^*X]/N$, and thus the expected number of visited nodes under $\rho I-\frac{X^*X}{N}$ is very close to  the number under $\rho_{E} I-\frac{E[X^*X]}{N}$. In fact, $\frac{(X^*X)_{i,j}}{N}$ can be written as the sum of independent random variables:
\begin{align}\label{eq:factarizationofX}
\frac{(X^*X)_{i,j}}{N}&=\frac{(\s^*_{i}\h+\w_{i})^*(\s^*_{j}\h+\w_{j})}{N}\notag\\
&=\frac{\sum\limits_{k=1}^{N}(\s^*_{i}h_k+w_{k,i})^*(\s^*_{j}h_{k}+w_{k,j})}{N}\notag\\
%&=\s_{i}\s^*_{j}\frac{\h^*\h}{N}+\frac{\w^*_{i}\w_{j}}{N}+\frac{\h^*\s_{i}\w_{j}}{N}+
%\frac{\w^*_{i}\h\s^*_{j}}{N}\notag\\
&=\s_{i}\s^*_{j}\frac{\sum^{N}_{k=1} h^*_{k}h_{k}}{N}+\frac{\sum^{N}_{k=1} w^*_{k,i} w_{k,j}}{N}\notag\\
&+\frac{\s_{i} \sum^{N}_{k=1} h^*_{k} w_{k,j}}{N}+\frac{\s^*_{j} \sum^{N}_{k=1} w^*_{k,i} h_{k}}{N}, \\
\end{align}
where $\w_{i}$ is the $i$-th column of $W$. Then we can find the expectation and the variance of (\ref{eq:factarizationofX}) as follows:
\begin{align}\label{eq:EXP}
E[\frac{(X^*X)_{i,j}}{N}]&= \s_{i}\s^*_{j}\frac{\sum^{N}_{k=1}  E(h^*_{k}h_{k})}{N}+\frac{\sum^{N}_{k=1} E(w^*_{k,i} w_{k,j})}{N}\notag\\
&+\frac{\s_{i} \sum^{N}_{k=1} E(h^*_{k} w_{k,j})}{N}+\frac{\s^*_{j} \sum^{N}_{k=1} E(w^*_{k,i} h_{k})}{N}, \notag\\
&=\begin{cases}
    1+\sigma^2_{w},  & \text{if } i= j\\
    \s_{i}\s^*_{j},  & \text{otherwise}
\end{cases}%\s_{i}\s^*_{j}+\frac{\sum^{N}_{\substack{k=1\\i=j}} var(w_{k})}{N}
\end{align}

$var(\frac{(X^*X)_{i,j}}{N})=\begin{cases}
    (1+2\sigma^2_{w}+\sigma^4_{w})/N,  & \text{if } i= j\\
    (2+2\sigma^2_{w}+\sigma^4_{w})/N,  & \text{otherwise}
\end{cases}$

The weak law of large numbers states that the sample mean of a random variable converges to its expectation in probability. Thus, for any pair $1\leq i,j \leq N$, for any $\xi>0$ and $\epsilon>0$, as $N$ goes to infinity, we have
\begin{equation}
P(|\frac{(X^*X)_{i,j}}{N}- \frac{E[(X^*X)_{i,j}]}{N}| \geq\varepsilon)\leq \xi.
\end{equation}

This means that, for any $\xi>0$ and $\epsilon>0$, as $N$ goes to infinity, we have
\begin{equation}
P(\|\frac{X^*X}{N}- \frac{E[X^*X]}{N}\|_F \leq \varepsilon)\geq 1-\xi,
\end{equation}
where $\|\cdot\|_{F}$ is the Frobenius norm.

Since $\rho$ is the maximum eigenvalue of $\frac{X^*X}{N}$, by the triangular inequality for the spectral norm
$$|\rho-\rho_E| <\|\frac{X^*X}{N}- \frac{E[X^*X]}{N}\|_2.$$

Since $$\|\frac{X^*X}{N}- \frac{E[X^*X]}{N}\|_2\leq \|\frac{X^*X}{N}- \frac{E[X^*X]}{N}\|_F,$$
we have
$$|\rho-\rho_E| <\|\frac{X^*X}{N}- \frac{E[X^*X]}{N}\|_F \leq \epsilon,$$
with probability at least $1-\xi$, as $N\rightarrow \infty$.

Using the triangular inequality for the spectral norm and the Frobenius norm, we have
$$\|\rho I -\frac{X^*X}{N}- (\rho_{E}I-\frac{E[X^*X]}{N})\|_2\leq 2\epsilon,$$
and
$$\|\rho I -\frac{X^*X}{N}- (\rho_{E}I-\frac{E[X^*X]}{N})\|_F\leq (\sqrt{T}+1)\epsilon,$$
with probability at least $1-\xi$, as $N\rightarrow \infty$.

Now since the Cholesky decomposition of  $(\rho I -\frac{X^*X}{N})$ is continuous at the point
$A=\rho_{E}I-\frac{E[X^*X]}{N}$,  for any $\epsilon>0$ and $\xi>0$, as $N$ goes to infinity,
$$\|R-\grave{R}\|_F \leq \epsilon$$
holds true with probability at least $1-\xi$. Thus as $N$ goes to infinity, for any full-length sequence $\widetilde{\s}^*$, with probability at least $1-\xi$,
$$|M_{\widetilde{\s}^*_{i:T}}^{\grave{R}}-M^{R}_{\widetilde{\s}^*_{i:T}}|=|\widetilde{\s}^*(R_{i:T}-\grave{R}_{i:T})\widetilde{\s}| \leq \|\widetilde{\s}\|^2 \|R-\grave{R}\|_F,$$
which is no bigger than $\|\widetilde{\s}\|^2 \epsilon$. Note here the superscripts $R$ and $R'$ in $M_{\widetilde{\s}^*_{i:T}}^{\grave{R}}-M^{R}_{\widetilde{\s}^*_{i:T}}$ describe what upper triangular matrix is used in calculating the metric.

Since we can take $\epsilon$ to be arbitrarily small, this means that, for a small enough $\epsilon$, the number of visited nodes per layer will also be equal to $|\Omega|$ under  $(\rho I -\frac{X^*X}{N})$, with probability at least $(1-\xi)$.
For a small enough constant $\epsilon>0$ and any constant $\xi>0$, as $N$ goes to infinity, the expected number of visited nodes at layer $i$  is upper bounded by
$$|\Omega|+(1-\xi) |\Omega|^{T-i},$$
since the largest number of visited nodes at layer $i$ when $r=\infty$  is $|\Omega|^{T-i}$.

If we take $\xi>0$ to be arbitrarily small, the expected number of visited nodes at layer $i$ will approach $|\Omega|$.

\end{proof}

\section{Simulation Results}
\label{sec:simulation}
 In this section, we simulate the performance and complexity of the exact ML algorithm for SIMO systems with a large number of receive antennas. We use the 4-QAM constant modulus constellation. Channel matrix entries are generated as i.i.d complex Gaussian random variables. We investigate the performance of the ML algorithm for $N$= $10$, $50$  $100$, and $500$ receive antennas. Two different data length values are examined, namely $T=8$ and $20$. We compare the performance of the joint ML non-coherent data detection algorithm with sub-optimal iterative and non-iterative channel estimation and data detection schemes.
We use least square (LS) and minimum mean square error (MMSE) channel estimation for the iterative and non-iterative detection schemes (the reader may refer to \cite{LS&MMSE} for the LS and MMSE channel estimation).

 We embed one symbol which is known by the receiver to resolve the phase ambiguity of the channel, at layer $T$ of the data sequence. In the non-iterative channel estimation case, the receiver estimates the channel vector using this training symbol. Then, the receiver uses this estimated channel vector  to detect the remaining $T-1$ transmitted symbols. The iterative channel estimation scheme exploits the detected data vector from the pervious iteration to obtain a new channel estimation, which, in turn, is used for data detection in the current iteration.  The iterative joint channel estimation and data detection scheme runs 100 iterations for each channel coherence block.

 In Figures \ref{LS8}, \ref{LS20} \ref{MMSE8} and \ref{MMSE20}, the symbol error rate (SER) of the ML algorithm has been evaluated as a function of SNR for $T=8$ and $20$ respectively, along with the SER of data detection based on the iterative and non-iterative LS and MMSE channel estimations. It can be seen that the ML algorithm outperforms the LS and MMSE iterative and non-iterative channel estimation schemes. For example, from Figures \ref{LS8} and \ref{MMSE8}, we see more than 2 dB improvement over the iterative channel estimation, and 3 dB improvement over the non-iterative channel estimation and data detection for $N=100$, at $10^{-2}$ SER. In Figures \ref{LS20} and \ref{MMSE20}, the ML detector provides a performance improvement of 2 dB over the iterative scheme and 4.5 dB improvement over the non-iterative scheme,at $10^{-2}$ SER.  One can notice  the improvement in the performance of ML channel estimation and data detection when we increase the number of receive antennas. From Figure \ref{MMSE20}, there is 2dB improvement SER for $N=100$ compared with $N=50$, whereas it is 7 dB using $N=100$ compared with $N=10$, at $10^{-1}$ SER.

 The complexity of the ML algorithm is evaluated based on the average number of nodes which are visited each layer during the algorithm execution.  In Figure \ref{proposed radius} we obtain the average number of visited nodes per layer for $T=20$ , $N=100$, and $N=500$ at SNR=-2dB.  This experiment is for the  4-QAM constellation,  using our proposed search radius satisfying $r^2=\frac{T}{8}$. For $N=100$, the average number of visited nodes per layer is already very low.  When $N=500$, the number of average visited nodes per layer is steady at a constant number, namely 4.  When the number of receive antennas goes from  $N=100$ to $N=500$, the simulation results show a clear reduction in the tree search complexity. Also, under a sufficiently large number of receive antennas, on average the joint ML algorithm will visit each layer 4 times, which is equal to the cardinality of the 4-QAM constellation. This is consistent with the theoretical prediction of our Theorem \ref{thm:ltt1}.

\section{Summary and Discussion}
\label{sec:conclusion} This paper shows, for the first time, the performance of joint ML channel estimation and data detection algorithm of massive SIMO wireless systems. We have shown that, as the number of receive antennas grows to infinity, the number of visited nodes per layer reaches a constant. Simulation results show that ML algorithm has better performance than iterative and non-iterative LS and MMSE channel estimation schemes. In addition, our simulation results verify our theorem by showing that the number of visited points per layer is equal to a constant number as  the number of receive antennas is sufficiently large.

\section{Appendix}

\begin{lemma}
Let $\s^*$ be the transmitted data sequence. Let us consider using $\rho_{E}I-\frac{E[X^*X]}{N}$ for calculating the sequence metric.  For any $ \widetilde{\s}^*$ such that $\widetilde{\s}^*\neq \s^*$, $M_{\widetilde{\s}^*}\neq 0$.
\label{thm:ltt2}
\end{lemma}
\begin{proof}[Proof of Theorem\ref{thm:ltt2}]
For any $ \widetilde{\s}^* \neq \s^*$, let $i$ be the closest integer to $T$ such that $\s^*_i \neq \widetilde{\s}^*_{i}$, where $1 \leq i \leq T-1$.  Then we can find the metric of $\widetilde{\s}^*_{i:T}$ based on (\ref{eq:Mertici1})
\begin{align}\label{eq:Mertici2}
M_{\widetilde{\s}^*_{i:T}} &= \|\sum^{T}_{k=i} L_{i,k} \widetilde{\s}_k\|^2+M_{\widetilde{\s}^*_{i+1:T}}\notag\\
  &=\|\sum_{k=i+1}^{T} L_{i,k} \s_k+L_{i,i} \widetilde{\s}_i\|^2,\notag\\
\end{align}
where $\widetilde{\s}^*_{i+1:T}=\s^*_{i+1:T}$, and $M_{\widetilde{\s}^*_{i+1:T}}=M_{{\s}^*_{i+1:T}}=0$ as proved in Theorem \ref{thm:ltt1}. Now we can write (\ref{eq:Mertici2}) as
\begin{align}
M_{\widetilde{\s}^*_{i:T}} &= \|\sum_{k=i}^{T} L_{i,k} \s_k- L_{i,i} {\s}_i+ L_{i,i} \widetilde{\s}_i\|^2\notag\\
  &=\|- L_{i,i} {\s}_i+ L_{i,i} \widetilde{\s}_i\|^2\notag\\
  &=\|L_{i,i} ({\widetilde{\s}_i-\s}_i)\|^2,\notag\\ \nonumber
\end{align}
where we have used the fact that $\sum_{k=i}^{T} L_{i,k} \s_k=0$,  proved in the proof of Theorem 3.1. Since ${\widetilde{\s}_i-\s}_i \neq 0$ by assumption, and $L_{i,i}\neq 0$ for $i\neq T$ according to Lemma \ref{thm:ltt3}, $M_{\widetilde{\s}^*_{i:T}}$ will not be zero as well.
\end{proof}

\begin{lemma}
$L_{i,i}\neq0$ for any $1\leq i \leq T-1$, and $L_{T,T}$ is equal to zero.
\label{thm:ltt3}
\end{lemma}

\begin{proof}[Proof of lemma \ref{thm:ltt3}]
$L_{ii}$ can be written as
\begin{align}\label{Livalue}
		L_{i,i}&=\sqrt {(T-1)-\sum_{j=1}^{i-1} \frac{T}{(T-(j-1))(T-j)}}\notag\\
			   &=\sqrt {(T-1)+\sum_{j=1}^{i-1} \left( \frac{T}{(T-(j-1))}-\frac{T}{(T-j)}\right)} \notag\\
               &=\sqrt{T-\frac{T}{T-(i-1)}}.\notag\\
\end{align}
 When $i=T$, (\ref{Livalue}) will be
\begin{align}
		L_{i,i}&=\sqrt{T-\frac{T}{T-(T-1)}}\notag\\
			   &=0.\notag\\\nonumber
\end{align}
It is also obvious that $L_{i,i}\neq0$ for any $i<$T.

\end{proof}

% Below is an example of how to insert images. Delete the ``\vspace'' line,
% uncomment the preceding line ``\centerline...'' and replace ``imageX.ps''
% with a suitable PostScript file name.
% -------------------------------------------------------------------------
\begin{figure}[!htb]
\centering
\includegraphics[width =3.0 in]{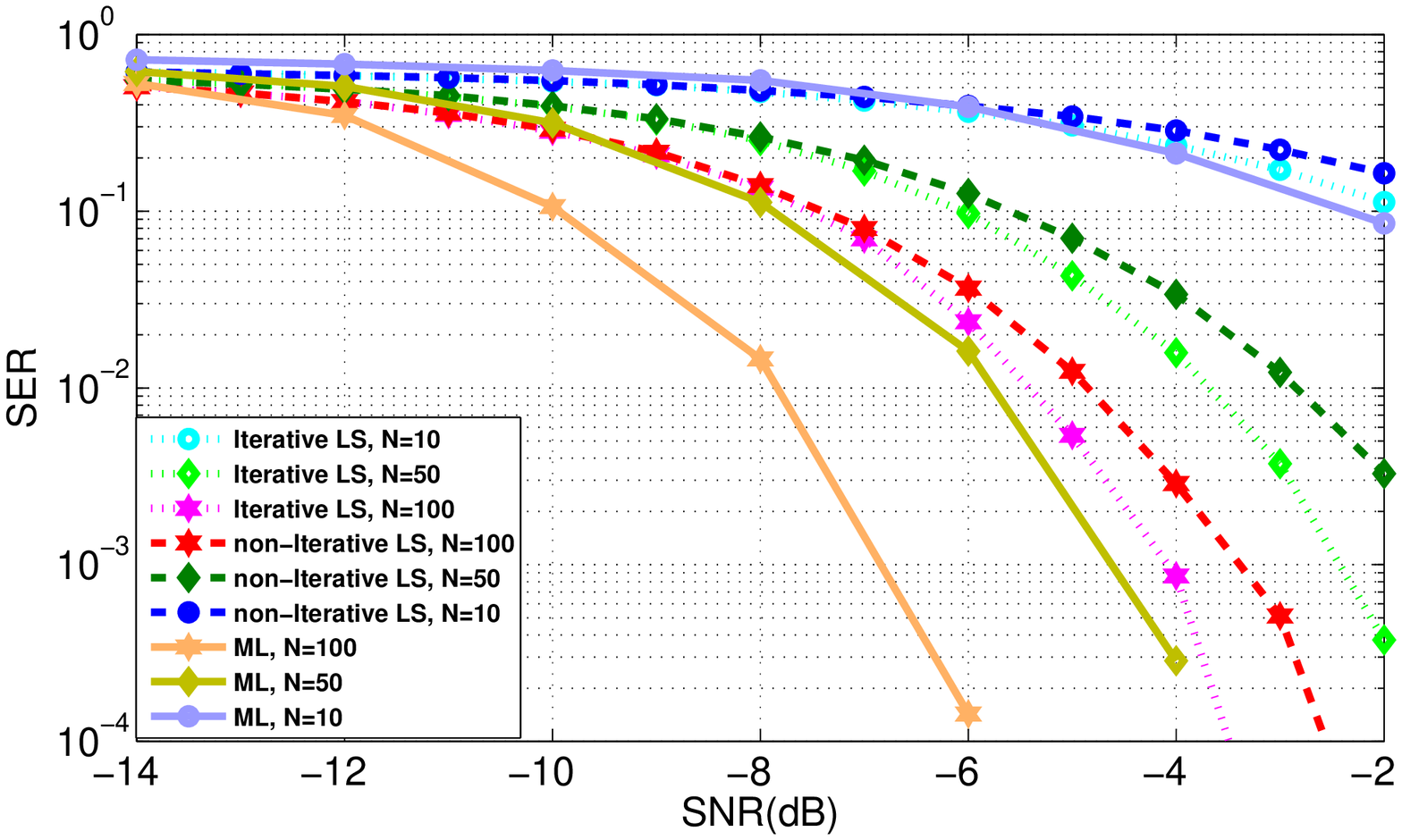}
\caption{SER vs SNR for joint ML channel estimation and data detection, iterative and non-iterative LS channel estimation for $T=8$}
\label{LS8}
\end{figure}

\begin{figure}[!htb]
\centering
\includegraphics[width =3.0 in]{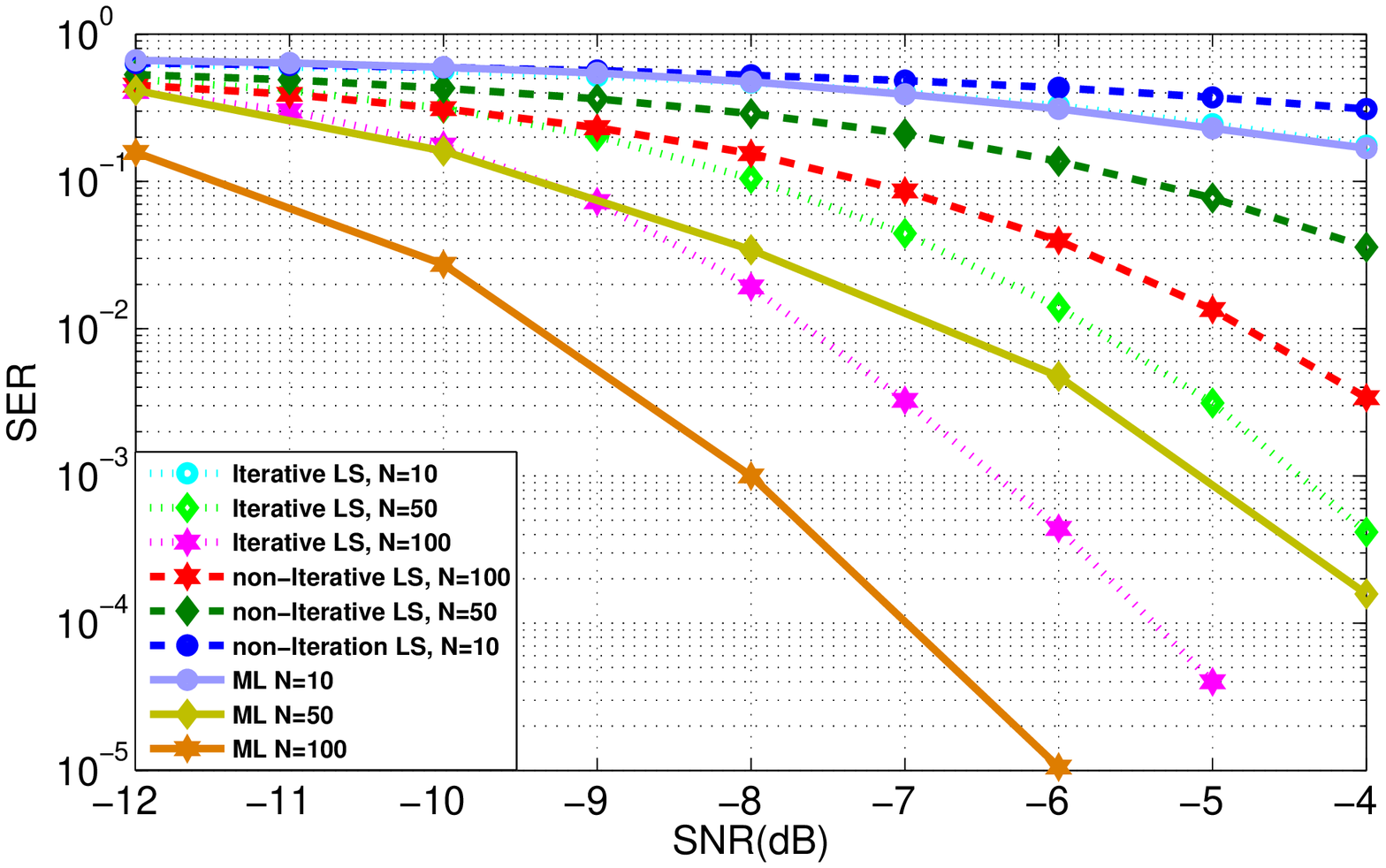}
\caption{SER vs SNR for joint ML channel estimation and data detection, iterative and non-iterative LS channel estimation for $T=20$}
\label{LS20}
\end{figure}

\begin{figure}[!htb]
\centering
\includegraphics[width =3.0 in]{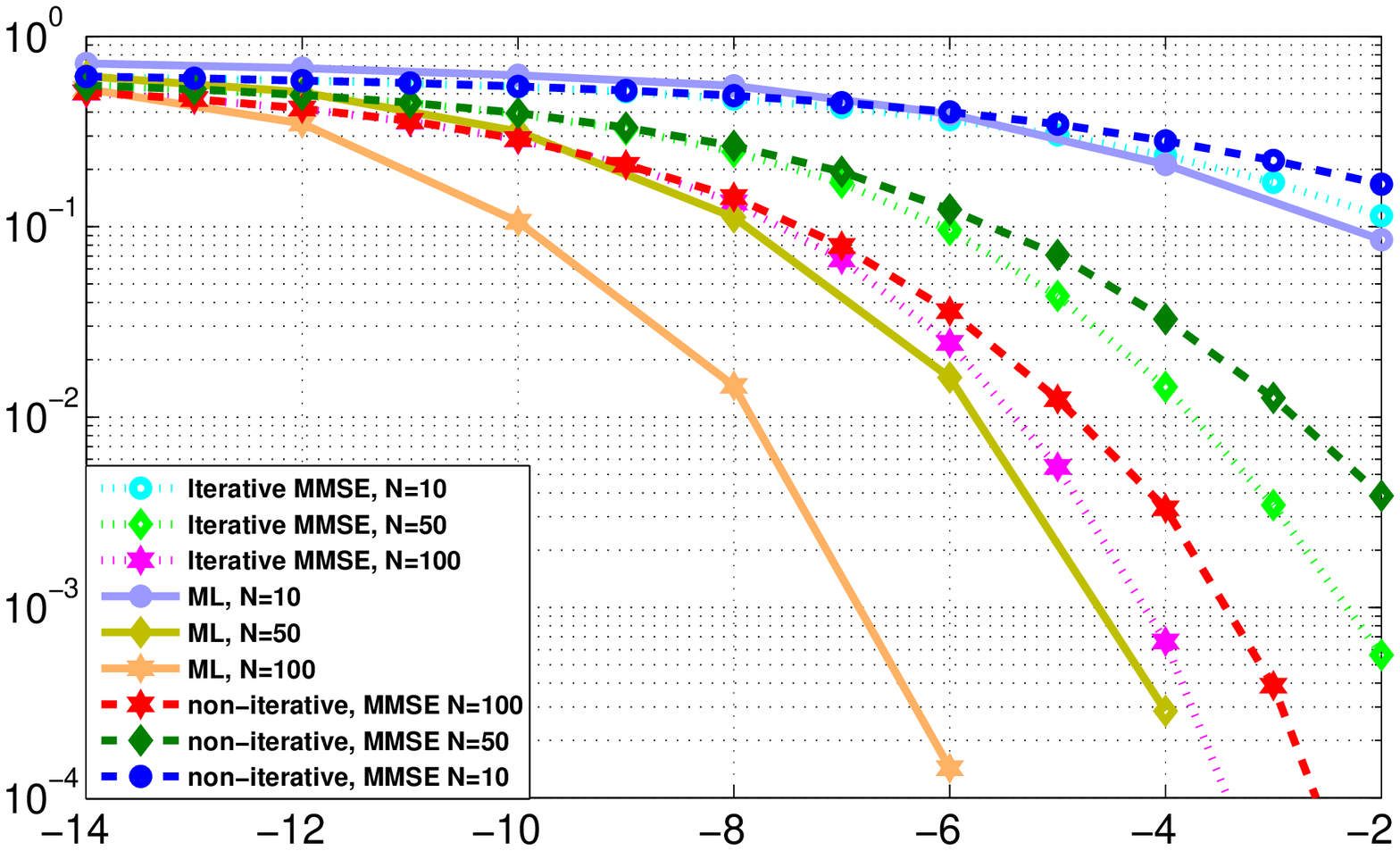}
\caption{SER vs SNR for joint ML channel estimation and data detection, iterative and non-iterative MMSE channel estimation for $T=8$} \label{MMSE8}
\end{figure}

\begin{figure}[!htb]
\centering
\includegraphics[width =3.0 in]{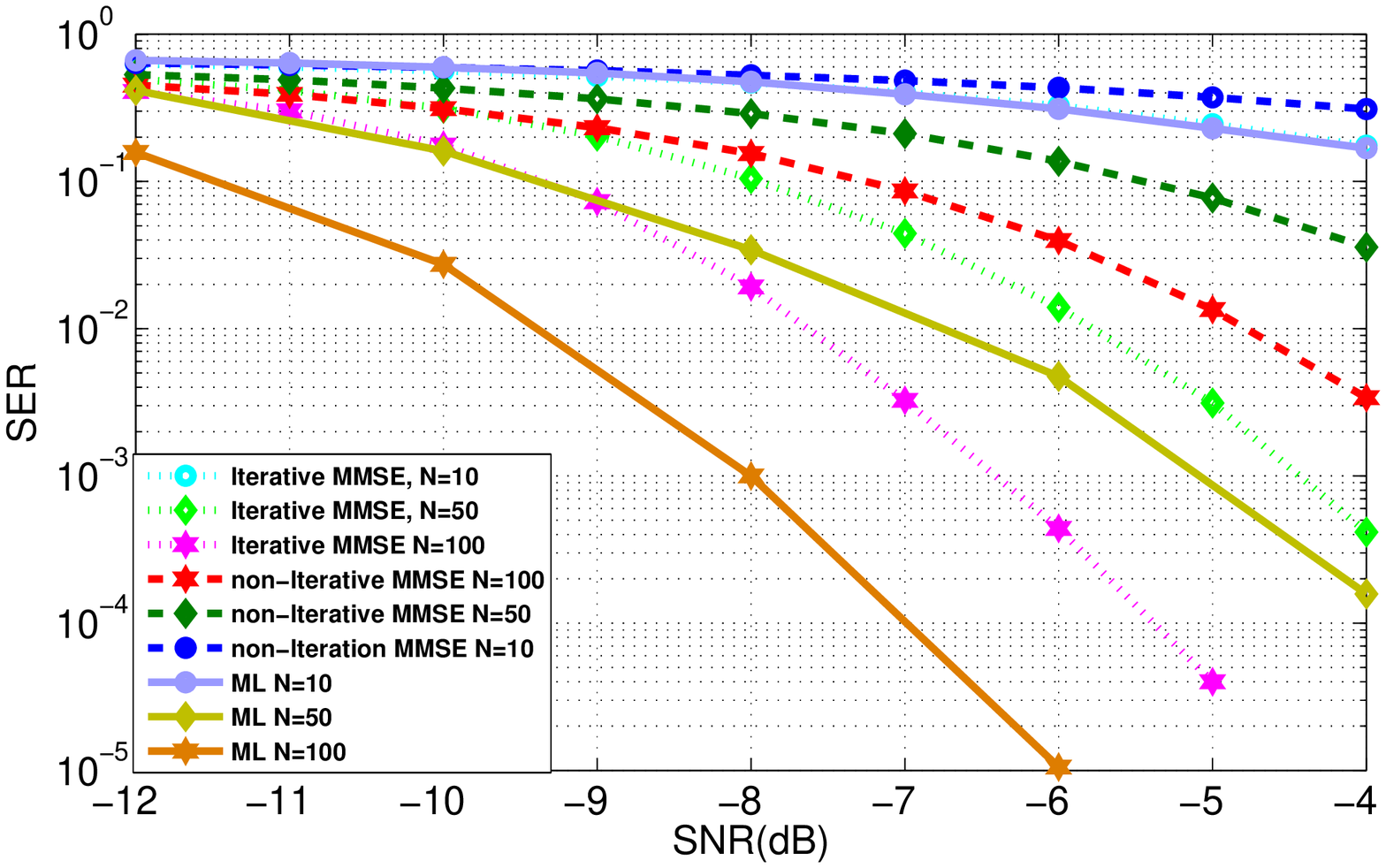}
\caption{SER vs SNR for joint ML channel estimation and data detection, iterative and non-iterative MMSE channel estimation for $T=20$} \label{MMSE20}
\end{figure}

\begin{figure}[!htb]
  \centering
  % Requires \usepackage{graphicx}
  \includegraphics[width =3.0 in]{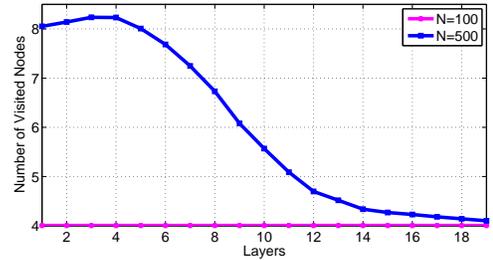}\\
  \caption{The average number of visited points per layer for $N = 100~\text{and}~500$, $T$=20 and SNR=-2dB. }\label{proposed radius}
\end{figure}

%%%%%%%%%%%%%%%%%%%%
%\section*{\centerline {\normalsize APPENDIX}}


\begin{thebibliography}{1}
\bibitem{Thomas}
Thomas L. Marzetta, ``Noncooperative Cellular Wireless with
Unlimited Numbers of Base Station Antennas,'' \emph{IEEE Transaction on Wireless Comunication.}, vol 9, no 11, pp.~3590-3600, Nov. 2003.
\bibitem{ScalingMIMO}
F. Rusek, D. Persson, B. K. Lau, E. G. Larsson, T. L. Marzetta, O. Edfors and F. Tufvesson, ``Scaling up MIMO: Opportunities and challenges with very large arrays,'' \emph{IEEE Signal Process. Mag.}, vol. 30, pp.~40-46, Jan 2013.
\bibitem{Challenges}
Lu Lu, G. Li, A. Swindlehurst, Alexei Ashikhmin and Rui Zhang. ``An Overview of Massive MIMO: Benefits and Challenges, '' \emph{ IEEE Journal of Selected Topics in Signal Processing.}, 2014.
\bibitem{NextGeneration}
Larsson, E. , Edfors, O. , Tufvesson, F. and Marzetta, T. ``Massive MIMO for next generation wireless systems,'' \emph{Communications Magazine, IEEE}, vol. 52, pp. 186-195, Feb 2014.
\bibitem{MUMIMO}
Hien Quoc Ngo, Larsson, E.G. and Marzetta, T.L. "Massive MU-MIMO downlink TDD systems with linear precoding and downlink pilots",  Communication, Control, and Computing (Allerton), 2013 51st Annual Allerton Conference on, pp. ~293-298.
\bibitem{HowmanyAntenna}
Jakob Hoydis, Stephan Ten Brink and Mérouane Debbah. "Massive MIMO: How many antennas do we need? " \emph{Communication, Control, and Computing (Allerton)}, 2011 49th Annual Allerton Conference on. IEEE, 2011.
\bibitem{Stoica}
P. Stoica and G. Ganesan, ``Space-time block codes: Trained, blind,
and semi-blind detection, '' \emph{Digital Signal Process.}, vol. 13,
pp.~93-105, 2003.
\bibitem{HarisConference}
P.Stoica, H.Vikalo and B.Hassibi,``Joint maximum-likelihood channel
estimation and signal detection for SIMO channels,'' in
\emph{Proceedings of 2003 International Conference on
Acoustics,speech and signal processing}, vol.4,2003,pp.~13--16.
\bibitem{Ma}
W.-K. Ma, B.-N. Vo, T.N. Davidson, and P.C. Ching, ``Blind ML
detection of orthogonal space-time block codes: High-performance,
efficient implementations," \emph{IEEE Transactions on Signal
Processing}, vol. 54, no. 2, pp.~738-751, Feb. 2006.
\bibitem{Swindlehurst}
A. L.Swindlehurst and G. Leus, ``Blind and semi-blind equalization
for generalized space-time block codes,'' \emph{IEEE Trans. Signal
Process.}, vol. 50, no. 10, pp.~2589-2498, 2002.
\bibitem{Manton}
S. Shahbazpanahi, A. B. Gershman, and J. H. Manton, ``Closed-form
blind MIMO channel estimation for orthogonal space-time block
codes,'' \emph{IEEE Trans. Signal Process}., vol. 53, no. 12, pp.
4506-4517, Dec. 2005.
\bibitem{Hassibi}
H. Vikalo, B. Hassibi and P. Stoica, ``Efficient joint
maximum-likelihood channel estimation and signal detection,''
\emph{IEEE Transactions on Wireless Communications}, vol 5, no 7, pp.~1838-1845, Jul 2006.
\bibitem{Weiyu}
Weiyu Xu, Stojnic M. and Hassibi B., " Low-complexity blind maximum-likelihood detection for SIMO systems with general constellations ",  Acoustics, Speech and Signal Processing, 2008. ICASSP 2008. IEEE International Conference on, pp. 2817 - 2820.
\bibitem{Andera}
Mainak Chowdhury, Alexandros Manolakos and Andrea J. Goldsmith, " Design and performance of noncoherent massive SIMO systems." CISS 2014: 1-6.
\bibitem{LS&MMSE}
Mehrzad Biguesh and Alex B. Gershamn, ``Training-Based MIMO channel estimation: a study of estimator tradeoffs and optimal training signals,'' \emph{IEEE Transactions on Signal Processing}, vol.54,  pp.~884-893,  March 2006.
\bibitem{Dimitris}
Dimitris S. Papailiopoulos, and George N. Karystinos ,``Maximum-Likelihood noncoherent OSTBC detection with polynomial complexity,'' \emph{IEEE Transactions on Wireless Communications
}, vol. 9, no. 6, June 2010.
\bibitem{NumbericalMathmatics}
A. Quarteroni, R. Sacco, and F. Saleri, Numerical Mathematics,  2000 :Springer-Verlag
\bibitem{SphereComplexity}
 Babak Hassibi, and Haris Vikalo, `` On the Sphere-Decoding Algorithm I. Expected Complexity,'' \emph{IEEE  Transactions on Signal Processing}, vol. 53, no. 8, August 2005.
\bibitem{OFDMBlind}
Al-Naffouri, T.Y.; Dahman, A.A.; Sohail, M.S.; Weiyu Xu; Hassibi, B., ``Low-Complexity Blind Equalization for OFDM Systems With General Constellations," \emph{IEEE Transactions on~~ Signal Processing} , vol.60, no.12, pp.6395-6407, 2012.
\end{thebibliography}
\end{document}